\documentclass[11pt]{article}
\pdfoutput=1
\usepackage{amsmath,amsfonts,amssymb,diagrams}
\usepackage{graphicx}
\usepackage{fullpage}

\newtheorem{theorem}{Theorem}[section]

\newtheorem{proposition}[theorem]{Proposition}

\newenvironment{proof}[1][Proof]{\begin{trivlist}
\item[\hskip \labelsep {\bfseries #1}]}{\end{trivlist}}
\newenvironment{definition}[1][Definition]{\begin{trivlist}
\item[\hskip \labelsep {\bfseries #1}]}{\end{trivlist}}
\newenvironment{example}[1][Example]{\begin{trivlist}
\item[\hskip \labelsep {\bfseries #1}]}{\end{trivlist}}

\begin{document}
\title{Fuzzy Topological Systems\thanks{This paper was accepted for presentation and it was 
    read at the {\em 8th Panhellenic Logic Symposium}, July 4--8, 2011, Ioannina, Greece.}}
\author{Apostolos Syropoulos\\ 
        Greek Molecular Computing Group\\
        Xanthi, Greece\\
        \texttt{asyropoulos@gmail.com}\\
\and 
       Valeria de Paiva\\
       School of Computer Science\\
       University of Birmingham, UK\\
       \texttt{valeria.depaiva@gmail.com}\\
}
\date{}
\maketitle
\begin{abstract}
Dialectica categories are a very versatile categorical model of linear logic. 
These have been used to model many seemingly different things (e.g.,
Petri nets and Lambek's calculus). In this note, we expand our previous
work on fuzzy petri nets to deal with fuzzy topological systems. One basic idea is 
to use as the dualizing object in the Dialectica categories construction, 
the unit real interval $\mathrm{I}=[0,1]$,  which has all the properties of a 
{\em lineale}. The second basic idea is to generalize Vickers's notion of a 
topological system.
\end{abstract}
\section{Introduction}
Fuzzy set theory and fuzzy logic have been invented by Lotfi ali Asker Zadeh. 
This is a theory that started from a generalization of the set concept and the 
notion of a truth value (for an overview, for example, see~\cite{klir95}). In fuzzy 
set theory, an element of a fuzzy subset belongs to it to a degree, which is usually 
a number between $0$ and $1$. For example, if we have a fuzzy subset of white colors, 
then all the gray-scale colors are white to a certain degree and, thus, belong to t
his set with a degree. The following definition by Zadeh himself explains what  fuzzy logic 
is:\footnote{The definition was posted to the bisc-group mailing list on 22/11/2008.}        
\begin{definition}
Fuzzy logic is a precise system of reasoning, deduction and computation in which 
the objects of discourse and analysis are associated with information which is, or 
is allowed to be, imprecise, uncertain, incomplete, unreliable, partially true or 
partially possible. 
\end{definition}

Categories, which were invented by Samuel Eilenberg and Saunders Mac Lane, form a very 
high-level abstract mathematical theory that unifies all branches of mathematics.
Category theory plays a central role in modern mathematics and 
theoretical computer science, and, in addition,  it is used in mathematical 
physics, in  software engineering, etc. Categories have been used to model and
study logical systems. In particular, the Dialectica categories of de Paiva~\cite{paiva89} 
are categorical model of linear logic~\cite{girard95}. These categories have been used
to model Petri nets~\cite{brown90}, the Lambek Calculus~\cite{paiva90}, 
state in programming~\cite{correaetal96}, and to define fuzzy petri 
nets~\cite{paivasyropoulos11}. Using some of the ideas in our previous work on 
fuzzy petri nets, we wanted to develop the idea of fuzzy topological systems, that is, 
the fuzzy counterpart of Vickers's~\cite{vickers90} topological systems. In this note, 
we present fuzzy topological systems and discuss some of their properties. 

\section{The category $\mathsf{Dial}_{\mathrm{I}}(\mathbf{Set})$}\label{basic:notions}
The Dialectica categories construction (see for example~\cite{paiva06}) can be  
instantiated using any lineale  and the basic category $\mathbf{Set}$.  As discussed 
in~\cite{syropoulos06}, the unit interval, since it is a Heyting algebra, has all the 
properties of a lineale structure. Recall that a lineale is a structure defined as follows:
\begin{definition}
The quintuple $(L, \le, \circ, 1, \multimap)$ is a lineale if:
\begin{itemize}
\item $(L,\le)$ is poset,
\item $\circ:L\times L\rightarrow L$ is an order-preserving multiplication,
      such that $(L,\circ, 1)$ is a symmetric monoidal structure (i.e.,
      for all $a\in L$, $a\circ 1=1\circ a=a$).
\item if for any $a, b\in L$ exists a largest $x\in L$ such that 
      $a\circ x\le b$, then this element is denoted $a\multimap b$ and
      is called the pseudo-complement of $a$ with respect to $b$.
\end{itemize}
\end{definition}
Now, one can prove that the quintuple 
$(\mathrm{I}, \le, \wedge, 1, \Rightarrow)$, where $\mathrm{I}$ is the unit interval, 
$a\wedge b=\min\{a,b\}$, and  $a\Rightarrow b=\bigvee\{c: c\wedge a\le b\}$ 
($a\vee b=\max\{a,b\}$),  is a lineale.

Let $U$ and $X$ be nonempty sets. A binary fuzzy relation $R$ in $U$ and $X$ is a 
fuzzy subset of $U\times X$, or $U\times X\to {\mathrm{I}}$. The value of $R(u,x)$ 
is interpreted as the {\em degree} of membership of the ordered pair $(u,x)$ in $R$. 
Let us now define a category of fuzzy relations.
 
\begin{definition}
The category $\mathsf{Dial}_{\mathrm{I}}(\mathbf{Set})$ has  as objects triples 
$A = (U, X, \alpha)$, where $U$ and $X$ are sets and $\alpha$ is a map 
$U\times X\to\mathrm{I}$. Thus, each object is  a fuzzy relation.  A map from 
$A=(U,X,\alpha)$ to $B=(V,Y,\beta)$  is a pair of  $\mathbf{Set}$ maps $(f,g)$,  
$f:U\to V$, $g:Y\to X$ such that
\begin{displaymath}
\alpha(u, g(y)) \leq \beta(f(u),y),
\end{displaymath}
or in pictorial form:
\begin{diagram}
U\times Y                         & \rTo^{\mathrm{id}_{U}\times g}     & U\times X\\
\dTo^{f\times\mathrm{id}_{Y}}     & \ge                                & \dTo_{\alpha}\\
V\times Y                         & \rTo_{\beta}                       & \mathrm{I}
\end{diagram}
\end{definition}

Assume that $(f,g)$ and $(f',g')$ are the following arrows:
\begin{displaymath}
(U,X,\alpha)\overset{(f,g)}{\underset{}{\longrightarrow}}
(V,Y,\beta)\overset{(f',g')}{\underset{}{\longrightarrow}}(W,Z,\gamma).
\end{displaymath}
Then $(f,g)\circ(f',g')=(f\circ f',g'\circ g)$ such that
\begin{displaymath}
\alpha\Bigl(u,\bigl( g'\circ g\bigr)(z)\Bigr) \le \gamma\Bigl(\bigl(f\circ f'\bigr)(u),z\Bigr).
\end{displaymath}

Tensor products  and the internal-hom in $\mathsf{Dial}_{\mathrm{I}}(\mathbf{Set})$ are given as  in the Girard-variant of the  Dialectica construction~\cite{paiva89}.
Given objects $A = (U, X, \alpha)$ and  $B = (V, Y, \beta)$, the tensor product  $A\otimes B$ is
$(U\times V, X^V\times Y^U, \alpha\times\beta)$, where the $\alpha\times\beta$ is the relation that, using the lineale structure of $I$, takes the minimum of the membership degrees. The linear function-space or internal-hom is 
given by $A\to B= (V^U\times Y^X, U\times X, \alpha\to \beta)$, where again the relation $\alpha\to \beta$ is given by the implication in the lineale. With this structure we obtain:
\begin{theorem}
The category ${\sf Dial}_{{\mathrm{I}}}({\bf Sets})$ is a  monoidal closed category with products and coproducts.
\end{theorem}
Products and coproducts are given by $A\times B= (U\times V, X+ Y, \gamma)$ and 
$A\oplus B= (U+ V, X\times Y, \delta)$, where 
$\gamma:U\times V\times (X+Y)\to \mathrm{I}$ is the 
fuzzy relation that is defined as follows
\begin{displaymath}
\gamma\bigl((u,v),z\bigr) = \left\{\begin{array}{ll}
                              \alpha(u,x),& \text{if $z=(x,0)$}\\
                              \beta(v,y), & \text{if $z=(y,1)$}
                             \end{array}\right.
\end{displaymath}
Similarly for the coproduct $A\oplus B$.

\section{Fuzzy Topological Systems}
Let $A = (U, X, \alpha)$ be an object of $\mathsf{Dial}_{\mathrm{I}}(\mathbf{Set})$, 
where $X$ is a frame, that is, a poset $(X, \leq)$ where
\begin{enumerate}
\item every subset $S$ of $X$ has a join
\item every finite subset $S$ of $X$ has a meet
\item binary meets distribute over joins, if $Y$ is a subset of $X$:
\begin{displaymath}
x\wedge\bigvee Y= \bigvee\Bigl\{ x\wedge y: y\in Y\Bigr\}.
\end{displaymath}
\end{enumerate}
Given such a triple, we can view $A$ as a {\em fuzzy topological system}, that is, the fuzzy 
counterpart of Vickers's~\cite{vickers90} {\em topological systems}.

A {\em topological system} in Vicker's monograph\cite{vickers90} is a triple
$(U,\models,X)$, where $X$ is a frame whose elements are called {\em opens} and 
$U$ is a set whose elements are called {\em points}. Also,  the relation  $\models$ 
is a subset of $U\times X$, and when $u\models x$, we say that $u$ {\em satisfies} 
$x$. In addition, the following must hold
\begin{itemize}
\item if $S$ is a finite subset of $X$, then
\begin{displaymath}
u\models \bigwedge S \Longleftrightarrow u\models x\;\text{for all $x\in S$}.
\end{displaymath}
\item if $S$ is any subset of $X$, then 
\begin{displaymath}
u\models\bigvee S \Longleftrightarrow u\models x\;\text{for some $x\in S$}.
\end{displaymath}
\end{itemize}

Given two topological systems $(U,X)$ and $(V,Y)$, a map from $(U,X)$ to $(V,Y)$ 
consists of a  function $f:U\rightarrow V$ and a frame homomorphism $\phi:Y\rightarrow X$,
if $u\models \phi(y) \Leftrightarrow f(u)\models y$. Topological systems and 
continuous maps between them form a category, which we write as $\mathbf{TopSystems}$.

In order to fuzzify topological systems, we need to fuzzify the relation ``$\models$.''
However, the requirement imposed on the relation of satisfaction is too severe when dealing
with fuzzy structures. Indeed, in some reasonable categorical models of fuzzy structures 
(see, for example~\cite{barr91,syropoulos06}), the authors use a weaker condition where the e
quivalence operator is replaced by an implication operator. Thus we suggest that the corresponding condition for morphisms of  fuzzy topological systems should become 
$u\models \phi(y) \Rightarrow f(u)\models y$.   
\begin{definition}
A {\em fuzzy topological system} is a triple $(U,\alpha,X)$, where $U$ is a set, $X$ is a frame
and $\alpha:U\times X\rightarrow\mathrm{I}$ a binary fuzzy relation such that:
\begin{itemize}
\item[(i)] If $S$ is a finite subset of $X$, then
\begin{displaymath}
\alpha(u, \bigwedge S)\le\alpha(u,x)\;\text{for all $x\in S$}.
\end{displaymath}
\item[(ii)] If $S$ is any subset of $X$, then 
\begin{displaymath}
\alpha(u,\bigvee S)\le\alpha(u,x)\;\text{for some $x\in S$}.
\end{displaymath}
\item[(iii)] $\alpha(u,\top)=1$ and $\alpha(u,\bot)=0$ for all $u\in U$.
\end{itemize} 
\end{definition}
To see that fuzzy topological systems also form a category we need to show that given morphisms $(f, F):(U,X)\to (V,Y)$ and  $(g, G):(V,Y)\to (W,Z)$, the obvious composition  $(g\circ f, F\circ G):(U,X)\to (W,Z)$ is also a morphism of fuzzy topological systems. But we know $\mathsf{Dial}_{\mathrm{I}}(\mathbf{Set})$ is a category and conditions (i), (ii) and (iii) do not apply to morphisms. Identities are given by $(id_U,id_X):(U,X)\to (U,X)$.

The collection of objects of $\mathsf{Dial}_{\mathrm{I}}(\mathbf{Set})$ that are fuzzy topological
systems and the arrows between them, form the category $\mathbf{FTopSystems}$, which is a
 subcategory of $\mathsf{Dial}_{\mathrm{I}}(\mathbf{Set})$.

\begin{proposition}
Any topological system $(U,X)$ is a fuzzy topological system $(U,\iota,X)$, where
\begin{displaymath}
\iota(u,x)=\left\{\begin{array}{ll}
                   1, & \text{when $u\models x$}\\
                   0, & \text{otherwise}
                  \end{array}\right.
\end{displaymath}
\end{proposition}
\begin{proof}
Consider the first property of the relation ``$\models$''
\begin{displaymath}
u\models \bigwedge S \Longleftrightarrow u\models x\;\text{for all $x\in S$}.
\end{displaymath}
This will be translated to 
\begin{displaymath}
\iota(u,\bigwedge S)\le \iota(u, x)\;\text{for all $x\in S$}.
\end{displaymath}
The inequality is in fact an equality since whenever $u\models x$, $\iota(u,x)=1$. Therefore,
we can transform this condition into the following one
\begin{displaymath}
\iota(u,\bigwedge S) = \iota(u, x)\;\text{for all $x\in S$}.
\end{displaymath}
A similar argument holds true for the second property.
\end{proof}
The following result is based on the previous one: 
\begin{theorem}
The category of topological systems is a full subcategory of 
$\mathsf{Dial}_{\mathrm{I}}(\mathbf{Set})$.
\end{theorem}

Obviously, it is not enough to provide generalization of structures---one needs to
demonstrate that these new structures have some usefulness. The following example
gives an interpretation of these structures in a ``real-life'' situation.
\begin{example}
Vickers~\cite[p.~53]{vickers90} gives an interesting physical interpretation of 
topological systems. In particular, he considers the set $U$ to be a set of programs 
that generate bit streams and the opens to be assertions about bit streams. For exanple, 
if $u$ is a program that generates the infinite bit stream 010101010101\ldots and 
``\textbf{starts} 01010'' is an assertion that is satified if a bit stream starts 
with the digits ``01010'', then this is expressed as follows: 
\begin{displaymath}
x\models \text{\textbf{starts} 01010}.
\end{displaymath}
Assume now that $x'$ is a program that produces bit streams that look like the following one
\begin{center}
\includegraphics[scale=1]{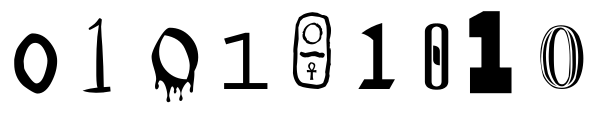}
\end{center}
The individuals bits are not identical to either ``1'' or ``0,'' but rather similar 
to these. One can speculate that these bits are the result of some interaction of 
$x'$ with its environment and this is the reason they are not identical. Then, we can 
say that $x'$ satisfies the assertion ``\textbf{starts} 01010'' to some degree, since 
the elements that make up the stream produced by $x'$ are not identical, but rather similar.
\end{example}
\section{From Fuzzy Topological Systems to Fuzzy Topological Spaces}

It is not difficult to map fuzzy topological systems to fuzzy topological spaces 
(for an overview of the theory of fuzzy topologies see~\cite{ying-ming97}). 
The following definition shows how to map an open to fuzzy set:
\begin{definition}
Assume that $a\in A$, where $(U,\alpha,X)$ is a fuzzy topological space. Then the 
{\em extent} of an open $x$ is a function whose graph is given below:
\begin{displaymath}
\Bigl\{ \bigl(u,\alpha(u,x) \bigr):u\in U\Bigr\}.
\end{displaymath}
\end{definition}
\begin{proposition}
The collection of all fuzzy sets created by the extents of the members of $A$ 
correspond to a fuzzy topology on $X$. 
\end{proposition}
\begin{proof}
Assume that $a$ and $b$ are opens and let $\mathbf{a}(x)=\alpha(x,a)$,
$\mathbf{b}(x)=\alpha(x,b)$, and $\psi(x)=\alpha(x,a\wedge b)$. Then 
$\alpha(x,a\wedge b)\le \alpha(x,a)$ and $\alpha(x,a\wedge b)\le \alpha(x,b)$. In
different words,  $\psi(x)\le\mathbf{a}(x)$ and $\psi(x)\le\mathbf{b}(x)$, which implies
that $\psi(x)\le\min\{\mathbf{a}(x),\mathbf{b}(x)\}$ that is $\psi=a\cap b$. 
Similarly, assume that $\{a_i\}$ is a collection of opens such that $\mathbf{a}_{i}(x)=\alpha(x,a_i)$
and $\phi(x)=\alpha(x,\bigvee_{i}a_i)$. The fact that there is one $\phi(x)\le\mathbf{a}_{j}(x)$,
while for all other $\mathbf{a}_{i}$ it holds that $\phi(x)\ge\mathbf{a}_{j}(x)$, implies that
$\phi(x)=\sup_{i}\mathbf{a}_{i}(x)$, that is, $\phi=\bigcup_{i}\mathbf{a}_i(x)$. Finally, the
last conditions generate the sets $\mathbf{1}(x)=1$ and $\mathbf{0}(x)=0$. So, the
opens form a fuzzy topology on $X$.
\end{proof}

\section{Products and Sums of Fuzzy Topological Systems}
In section~\ref{basic:notions} we described the categorical products and coproducts
of any two objects of $\mathsf{Dial}_{\mathrm{I}}(\mathbf{Set})$. Given two fuzzy 
topological systems $A=(U, X, \alpha)$ and $B= (V, Y, \beta)$, their topological
product is the space $A\times B= (U\times V, X+ Y, \gamma)$. Since $X$ and $Y$ are frames
it is necessary to modify the definition of $X+Y$ and, consequently, the definition of
$\gamma$. 
\begin{definition}
Assume that $A=(U, X, \alpha)$ and $B= (V, Y, \beta)$ are two fuzzy topological
systems. Then their topological product $A\times B$ is the system
$(U\times V,\gamma,X\otimes Y)$, where $X\otimes Y$ is the tensor product
of the two frames $X$ and $Y$ (see~\cite[pp.~80--85]{vickers90} for details)
and $\gamma$ is defined as follows:
\begin{displaymath}
\gamma\bigl((u,v),\bigvee_{i} x_{i}\otimes y_{i}\bigr)=
\max\bigl\{\alpha(u,x),\beta(v,y)\bigr\}.
\end{displaymath}
\end{definition}
Obviously, the topological product is not the same as the categorical product. 
The topological sum is more straigthtforward:
\begin{definition}
Assume that $A=(U, X, \alpha)$ and $B= (V, Y, \beta)$ are two fuzzy topological
systems. Then their topological sum $A + B$ is the system
$(U+V,\gamma,X\times Y)$, where $\gamma$ is defined as follows:
\begin{displaymath}
\gamma\bigl(z,(x,y)\bigr)=\left\{\begin{array}{ll}
                                 \alpha(u,x),& \text{if $z=(u,0)$}\\
                                 \beta(v,y), & \text{if $z=(v,1)$}
                                 \end{array}
                           \right.
\end{displaymath}
\end{definition}
Comparing the topological sum with the categorical sum reveals that they are identical.
\section{Conclusions}
We have simply started thinking about the possibilities of using
Dialectica-like models in the context of fuzzy topological
structures. Much remains to be done, in particular we would like to
see if a framework based on an implicational notion of morphism like
ours can cope with embedding several of the other notions of fuzzy
sets considered by Rodabaugh~\cite{rodabaugh1983}. Also seems likely 
that we could extend the work of Solovyov~\cite{solovyov2010} on 
variable-basis topological spaces using similar ideas. 

\end{document}